\documentclass[journal]{IEEEtran}
 \usepackage{caption}
\usepackage[utf8]{inputenc}
\usepackage{amsmath,amssymb,amsfonts}
\usepackage{microtype}

\usepackage{graphicx}
\usepackage{cite}
\usepackage{xcolor}
\usepackage{float}
\usepackage{bm}
\usepackage{hyperref}
\usepackage{multirow}
\usepackage{makecell}
\usepackage{tikz}
\usetikzlibrary{arrows.meta, positioning, shapes, calc}
\usepackage{algorithmic}
\usepackage{algorithm}
\usepackage{array}
\usepackage{mathdots}
 \usepackage{mathtools}
\usepackage{subcaption}
\usepackage{enumitem}
\usepackage{comment}
\usepackage{textcomp}
\usepackage{stfloats}
\usepackage{url}
\usepackage{verbatim}
\newtheorem{definition}{Definition}

\newtheorem{theorem}{Theorem}
\newtheorem{proposition}{Proposition}
\newtheorem{corollary}{Corollary}
\newtheorem{assumption}{Assumption}

\def\real{\mathbb{R}}

\newcommand{\cB}{\ensuremath{\mathcal{B}}}

\newcommand{\cI}{\ensuremath{\mathcal{I}}}

\newcommand{\cP}{\ensuremath{\mathcal{P}}}
\newcommand{\cQ}{\ensuremath{\mathcal{Q}}}

\newcommand{\cT}{\ensuremath{\mathcal{T}}}
\newcommand{\cU}{\ensuremath{\mathcal{U}}}

\newcommand{\cX}{\ensuremath{\mathcal{X}}}

\definecolor{ieeegreen}{RGB}{0,110,70}
\definecolor{ieeeblue}{RGB}{0,51,153}

\title{Harmonic Stability of Power Systems: A Control-Theoretic Definition and Assessment Criteria}
\author{%
  N.~Bhoir\textsuperscript{1},
  A.~Sarkar\textsuperscript{1},
  J.E.~Machado\textsuperscript{1},
  J.~Schiffer\textsuperscript{1,2}%
  \thanks{This work was supported by the German Federal Government, the Federal Ministry of Research, Technology and Space, and the State of Brandenburg within the framework of the joint project EIZ: Energy Innovation Center (project numbers 85056897 and 03SF0693A) with funds from the Structural Development Act (Strukturstärkungsgesetz) for coal-mining regions.}%
  \thanks{\textsuperscript{1}Control Systems and Network Control Technology Group, Brandenburg University of Technology Cottbus--Senftenberg, 03046 Cottbus, Germany \texttt{\{bhoir, sarkar, machadom, schiffer\}@b-tu.de}.}%
  \thanks{\textsuperscript{2}Fraunhofer IEG, Fraunhofer Research Institution for Energy Infrastructures and Geotechnologies IEG, 03046 Cottbus, Germany.}%
}

\begin{document}
\maketitle

\begin{abstract}
Harmonic interactions have become a defining dynamic stability phenomenon in converter-based power systems (CBPSs). However, a formalization of the notion of harmonic stability in the context of nonlinear dynamical systems is not available thus far. In this paper, we propose a definition of harmonic stability formulated as a combination of two standard stability notions, namely, bounded-input bounded-output (BIBO) and internal stability properties with respect to any nominal periodic trajectory.
Moreover, we rigorously show that the local harmonic stability of any nominal periodic trajectory of nonlinear CBPSs is implied by the stability properties of its linear time-periodic (LTP) approximation.
In addition, to enable computationally tractable stability assessment, we develop a 
framework based on harmonic state-space (HSS) representations of LTP models. In particular, we provide time-invariant linear matrix inequality (LMI) conditions to certify harmonic stability.
The proposed methodology is illustrated on a grid-following converter system, where the HSS-based analysis numerically confirms harmonic stability around the periodic operating trajectory.
\end{abstract}

\begin{IEEEkeywords}
Harmonic stability,
converter-based power systems,
linear time-periodic systems,
harmonic state-space,
bounded-input bounded-output stability.
\end{IEEEkeywords}

\section{Introduction}\label{sec:introduction}
\IEEEPARstart{T}{he} ongoing effort to decarbonize power systems is driving the global energy transition toward the large-scale deployment of renewable energy sources. 
Renewable energy sources are predominantly integrated into the electrical grid via power electronic converters, which increasingly replace conventional synchronous machines as the primary interfaces between generation units and the power network~\cite{hatziargyriou2021}. As a result, modern power systems are becoming progressively more converter-dominated, fundamentally altering their dynamic behavior~\cite{wang2019}.

Power electronic converters offer flexible and precise control of power flows, enabling independent regulation of active and reactive power, voltage magnitude, and frequency~\cite{CIGRE_TB909,wang2019}. Moreover, their fast digital control loops and high controllability further facilitate advanced grid-supporting functionalities such as synthetic inertia and voltage and frequency support~\cite{hatziargyriou2021}. 
However, the high switching frequencies, phase-locked loop (PLL) dynamics, and rotating-frame control structures based on the abc--dq transformation inherent to many power electronic converters introduce complex dynamic phenomena, including harmonic interactions and broadband frequency coupling~\cite{Ying1997}. These interactions can lead to sustained oscillatory behavior and reduced system damping, often interpreted as negative damping in converter-based power systems (CBPSs)~\cite{Mollerstedt2000out_of_control}. Such phenomena have been observed both in analytical studies and in practical grid deployments~\cite{Buchhagen} and may ultimately compromise system stability~\cite{Sun2011}. As a result, the notion of  \emph{harmonic stability} emerged as a distinct concept for characterizing the behavior of nonlinear CBPSs under periodic and harmonically distorted operating conditions~\cite{arrillaga2003power,wang2019}. 
Early studies of harmonic stability predominantly adopted a steady-state viewpoint, in which a system is considered harmonically stable if bounded harmonic inputs produce bounded steady-state harmonic responses~\cite{arrillaga2003power}. Such approaches were initially developed for the analysis of individual devices and subsystems and were later extended to interconnected power networks, enabling the assessment of harmonic propagation, resonance phenomena, and harmonic interactions at the system level 
~\cite{Ying1997,mollerstedt2000,wang2019,becker2024}.

Beyond the steady-state perspective, the dynamic behavior of nonlinear CBPSs has been studied using various complementary stability concepts and methods. Under balanced sinusoidal operating conditions, converter dynamics are commonly transformed into a synchronously rotating $(dq)$ reference frame and linearized about an equilibrium point, yielding linear time-invariant (LTI) models that form the basis of impedance- and transfer-function-based stability analyses~\cite{Sun2011,Agorreta2011,Amin2017,CIGRE_TB909}.
However, the synchronous transformation is typically constructed using the nominal or estimated fundamental component and therefore neglects the exact periodic structure of the steady-state trajectory. Consequently, harmonic components are not represented explicitly, and cross-frequency coupling effects are eliminated by construction. As a result, such LTI models cannot capture the generation of new harmonics or the cross-frequency coupling phenomena that arise in nonlinear CBPSs~\cite{Mollerstedt2000out_of_control,Guerrero2019,DeRua2020}.

Overcoming the inherent limitations of the above-mentioned LTI approximations, linear time-periodic (LTP) system representations were proposed in the literature to locally characterize the behavior of nonlinear CBPS dynamics around general periodic steady motions \cite{wereley_GNC1990,wereley1990analysis}. Such LTP representations arise from the linearization of nonlinear dynamics around a periodic operating trajectory and are capable of modeling converter-induced harmonic interactions \cite{wereley1990analysis,mollerstedt2000,Mollerstedt2000out_of_control}, thus providing a suitable framework for assessing stability in contexts of harmonically distorted operating conditions~\cite{ salis2018,Beloqui2024}.

 Within the LTP framework, several stability analysis approaches have been reported in the literature, including Floquet multipliers and characteristic exponents for periodic-orbit stability analysis~\cite{Bittanti1991, Chicone2006}, generalized Nyquist criteria and harmonic transfer function (HTF) representations~\cite{wereley_GNC1990,wereley1990analysis,mollerstedt2000,liao2022}, as well as harmonic state-space (HSS) formulations and associated eigenvalue-based analyses~\cite{Zhou2004HarmonicStateOperator,GeoffreyLove,Kwon2017HarmonicSS,salis2018, Hu2025}. Collectively, these approaches have substantially advanced the understanding of converter-driven harmonic phenomena and stability mechanisms~\cite{wang2019,Eckel2024}; however,  these works either adopt an orbital-stability (\emph{e.g.}, \cite{Bittanti1991,Chicone2006}) or a bounded-input bounded-output (BIBO) stability interpretation (\emph{e.g.}, \cite{Sun2011,Agorreta2011,Amin2017}), thus highlighting the need for a unified formal definition for characterizing the notion of harmonic stability from the viewpoint of dynamical systems theory. 
Moreover, these approaches predominantly assess stability through the associated LTP approximations rather than formally defining or assessing the harmonic stability of the underlying nonlinear CBPS dynamics~\cite{Bittanti1991,Chicone2006,mollerstedt2000,salis2018}.

Considering the aforementioned discussion, in this paper we propose a control-theoretic framework for the harmonic stability analysis of nonlinear CBPSs.  Our framework is based on LTP and HSS system representations and hence is capable of capturing harmonic interactions and cross-frequency coupling effects, which are neglected by conventional LTI models. In particular, we make the following contributions:

\begin{itemize}
\item[i)] We introduce a control-theoretic definition of harmonic stability for nonlinear CBPSs. 
This definition combines BIBO stability with local internal stability with respect to
a nominal periodic operating trajectory.

\item[ii)] We establish sufficient conditions under which local harmonic stability of nonlinear CBPSs can be certified with respect to any feasible nominal periodic operating trajectory through the analysis of the corresponding LTP approximations around it. The proposed conditions combine time-varying matrix inequalities for the LTP system with suitable smoothness assumptions on the nonlinear dynamics. 

\item[iii)] Based on finite-dimensional HSS approximations of LTP systems, we develop a computationally tractable approach for assessing local harmonic stability via linear matrix inequalities (LMIs) and establish conditions under which the resulting characterization is exact.

\end{itemize}
The remainder of this paper is organized as follows. The notation,
assumptions, and stability notions that form the technical basis of our developments appear in Section~\ref{sec:Preliminaries}. The proposed harmonic-stability framework for nonlinear CBPS dynamics are introduced in
Section~\ref{sec:HS_formalization}. The analytical tools through which  harmonic stability
is linked to the LTP linearization and to HSS representations is developed in
Section~\ref{sec:HS_assessment_tools}. A detailed numerical case study based on a grid-following converter is reported in Section~\ref{sec:case_study}, and concluding remarks are given in Section~\ref{sec:conclusions}.

\section{Notation and Technical Background}\label{sec:Preliminaries}

In this section, we summarize the notation that is adopted in the paper and introduce some stability notions for nonlinear non-autonomous systems relevant for our subsequent developments. 
\vspace{-0.3cm}
\subsection*{Notation}
The sets of real, complex, natural, and integer numbers are denoted by $\mathbb R$, $\mathbb {C}$, $\mathbb {N} $, and $\mathbb {Z} $, respectively, and $\mathbb R_{\ge 0}$ denotes the set of non-negative real numbers. The $n$-dimensional Euclidean space is denoted by $\mathbb{R}^n$. The identity matrix of size $n$ is denoted by $I_n$, the zero vector in $\mathbb{R}^n$ by $\mathbf{0}_n$, and the zero matrix of size $m \times n$ by $\mathbf{0}_{m \times n}$. For a matrix $A$, $A^\top$ and $A^\star$ denote its transpose and conjugate transpose, respectively, and $\|A\|$ denotes its induced spectral norm unless otherwise specified. 
For a matrix $A \in \real^{n \times n}$, we write $A \succ 0$ (resp., $A \succeq 0$) if $A$ is symmetric, positive-definite (resp., positive-semidefinite), i.e., for any $x \in \real^n\setminus\{\mathbf{0}_n\}$, $x^\top A x > 0$ (resp., $x^\top A x \geq 0$). The imaginary unit $\sqrt{-1}$ is denoted by $j$. The real and imaginary parts of a complex quantity are denoted by
$\Re(\cdot)$ and $\Im(\cdot)$, respectively. The notation $\mathrm{diag}(\cdot)$ denotes a diagonal matrix with the specified entries on its main diagonal, and $\mathrm{blkdiag}(\cdot)$ denotes a block-diagonal matrix with the specified matrix blocks. For a subset $A$ of a normed linear space, by $\mathrm{int} (A)$ we refer to the interior of $A$. We define an open ball of radius $r$ around a set $A$ by $B_r(A) \coloneqq \{x \in \real^n \mid \inf_{y \in A} \{\|x-y\|\} < r\}$. 
For a square matrix $A$, we denote by $\sigma(A)$ its spectrum, i.e., the set of all eigenvalues of $A$. For a symmetric matrix $A \in \real^{n \times n}$, we denote its maximum eigenvalue by $\lambda_{\max}(A)$. For a vector $x :=[x_1,\;x_2,\;\cdots,\;x_n]^\top\in \mathbb{R}^n$, $\|x\|$ is its Euclidean norm defined as $\|x\| := \sqrt{x_1^2 + x_2^2+\cdots+x_n^2}$, and its infinity norm $\Vert x \Vert_\infty$ is defined as $ \|x\|_{\infty} \coloneqq \sup_{i} |x_i|$. 
A function $f:\mathbb{R}\rightarrow \mathbb R$ is said to be $T$-periodic if $T>0$ is the smallest real number for which $f(t+T)=f(t)$ for all $t\in \mathbb{R}$. For a continuously differentiable vector field $(x,y)\mapsto f(x,u)$,  we denote by $\frac{\partial f(x,u)}{\partial x}$ and
$\frac{\partial f(x,u)}{\partial u}$ the Jacobian matrices of $f$ with respect
to $x$ and $u$, respectively; with a slight abuse of notation, the evaluation of these
Jacobians at any continuous trajectory $(x^\star(t),u^\star(t))$ are respectively  denoted by  $\frac{\partial f(x^\star(t),u^\star(t))}{\partial x}$ and
$\frac{\partial f(x^\star(t),u^\star(t))}{\partial u}$. 

\subsection*{Relevant Stability Notions}
Let $D_x := \{x \in \real^n:~ \|x\|\leq r_x\}$ and $D_u := \{u \in \real^m :~ \|u\|\leq r_u\}$ for some positive constants $r_x$ and $r_u$ and consider the following nonlinear  system in state-space form~\cite{khalil, YangNL2021,Cecati2022}:
\begin{subequations} \label{eq:system}
\begin{align}
\dot{x}(t) &= f(x(t), u(t)), \label{eq:state} \\
y(t) &= h(x(t)), \label{eq:output}
\end{align}
\end{subequations}
where $x(t)\in D_x$ is the state, $u(t)\in D_u$ the external input and $y(t)\in\mathbb{R}^p$ the output. The drift vector field is $f:D_x\times D_u\to\mathbb{R}^n$, with $f(\mathbf{0}_n,\mathbf{0}_m)=\mathbf{0}_n$, and the output map is $h:D_x\to\mathbb{R}^p$ with $h(\mathbf{0}_n) = \mathbf{0}_p$.

The following assumptions are imposed for the system~\eqref{eq:system} to ensure sufficient smoothness of its linearization.
\begin{assumption}\label{assum:f}
Consider the system~\eqref{eq:system}.
    \begin{enumerate}
        \item[(i)]  The map $f$ is continuously differentiable on $D_x\times D_u$.
        \item[(ii)] The output map $h$ is continuously differentiable on $D_x$.
    \end{enumerate}
    \hfill $\square$
\end{assumption}

Last, we recall some standard stability notions to describe boundedness and convergence properties for the system \eqref{eq:system}.
\begin{definition}[BIBO Stability \cite{khalil}]\label{sec:BIBO}
	The system \eqref{eq:system} is BIBO stable if for every bounded input $u(t)$, the output $y(t)$ remains bounded, i.e.,
	\begin{equation}
		\|u(t)\|_\infty < \infty \Rightarrow \|y(t)\|_\infty < \infty\;\forall~ t \in \mathbb{R}.
	\end{equation}

The system \eqref{eq:system} is \emph{small-signal} BIBO stable if there exists an $\epsilon > 0$, such that for every input $u(t)$ with $\|u(t)\|_\infty < \epsilon$, the output $y(t)$ is bounded, i.e.,
\begin{equation}
    \|u(t)\|_\infty < \epsilon \quad \Rightarrow \quad \|y(t)\|_\infty < \infty\; \forall~ t \in \mathbb{R}.
   \label{eq:local_bibo}
\end{equation}
\hfill $\square$
\end{definition}

\begin{definition}[Local Uniform Asymptotic Stability {\cite[Definition~4.4]{khalil}}]
For $u\equiv\mathbf{0}_m$, the equilibrium point $x=\mathbf{0}_n$ of the system~\eqref{eq:system} is said to be \emph{locally uniformly asymptotically stable} if, for every $\epsilon>0$, there exists $\delta=\delta(\epsilon)>0$, independent of the initial time $t_0$, such that
$\|x(t_0)\|<\delta$ implies $\|x(t)\|<\epsilon$ for all $t\ge t_0\ge0$, and there exists a positive constant $c>0$, independent of $t_0$, such that, for all $\|x(t_0)\|<c$, $x(t)\to\mathbf{0}_n$ as $t\to\infty$, uniformly in $t_0$; that is, for every $\eta>0$, there exists $T=T(\eta)>0$, independent of $t_0$, such that
$\|x(t)\|<\eta,\;\forall\, t\ge t_0+T(\eta),\;\forall\, \|x(t_0)\|<c.$
\hfill $\square$
\end{definition}

\section{Harmonic Stability in CBPSs: A Control-Theoretic Definition}\label{sec:HS_formalization}
In general, the dynamics of CBPSs can be modeled as a nonlinear, non-autonomous system of the form~\eqref{eq:system}, see \cite{YangNL2021,Cecati2022}. Thus, for the subsequent derivations, we consider CBPSs represented by \eqref{eq:system}. Building on the above notions of stability, we now define harmonic stability for nonlinear CBPSs.
To this end, we begin by defining a {\em nominal} periodic motion for the system~\eqref{eq:system}.

\begin{definition}[Nominal periodic motion] A triple $(u^\star(t),x^\star(t),y^\star(t))\in \mathrm{int}(D_u)\times \mathrm{int}(D_x)\times \mathbb{R}^p$ is a nominal periodic motion of the system~\eqref{eq:system} if \( u^\star(t)\) is a continuous, time-periodic input with period $T_1>0$, which gives rise to a corresponding bounded time-periodic trajectory \( x^\star(t) \) with period $T_2>0$, and a corresponding bounded time-periodic output \(y^\star(t)\) with period $T_3>0$ satisfying 
\begin{subequations}\label{eq:steady_state_eqs}
    \begin{align}
\dot x^\star(t) &= f\bigl(x^\star(t),u^\star(t)\bigr), \\
 y^\star(t) &= h\bigl(x^\star(t)\bigr).
\end{align}
\end{subequations} 
    \hfill $\square$
\end{definition}
Importantly, nonlinear CBPSs dynamics may generate additional harmonic or subharmonic frequency components through cross-coupling effects~\cite{mollerstedt2000,becker2024}. Thus, in general, the periods $T_1$, $T_2$, and $T_3$ need not coincide. For illustrative purposes, Fig.~\ref{fig:nominal_periodic_motion} shows the response of a forced van der Pol--Duffing oscillator~\cite{jordan1999nonlinear}. A realistic nonlinear CBPS case study is presented in Section~\ref{sec:case_study}.
\begin{figure}[!t]
    \centering
    \includegraphics[width=\linewidth]{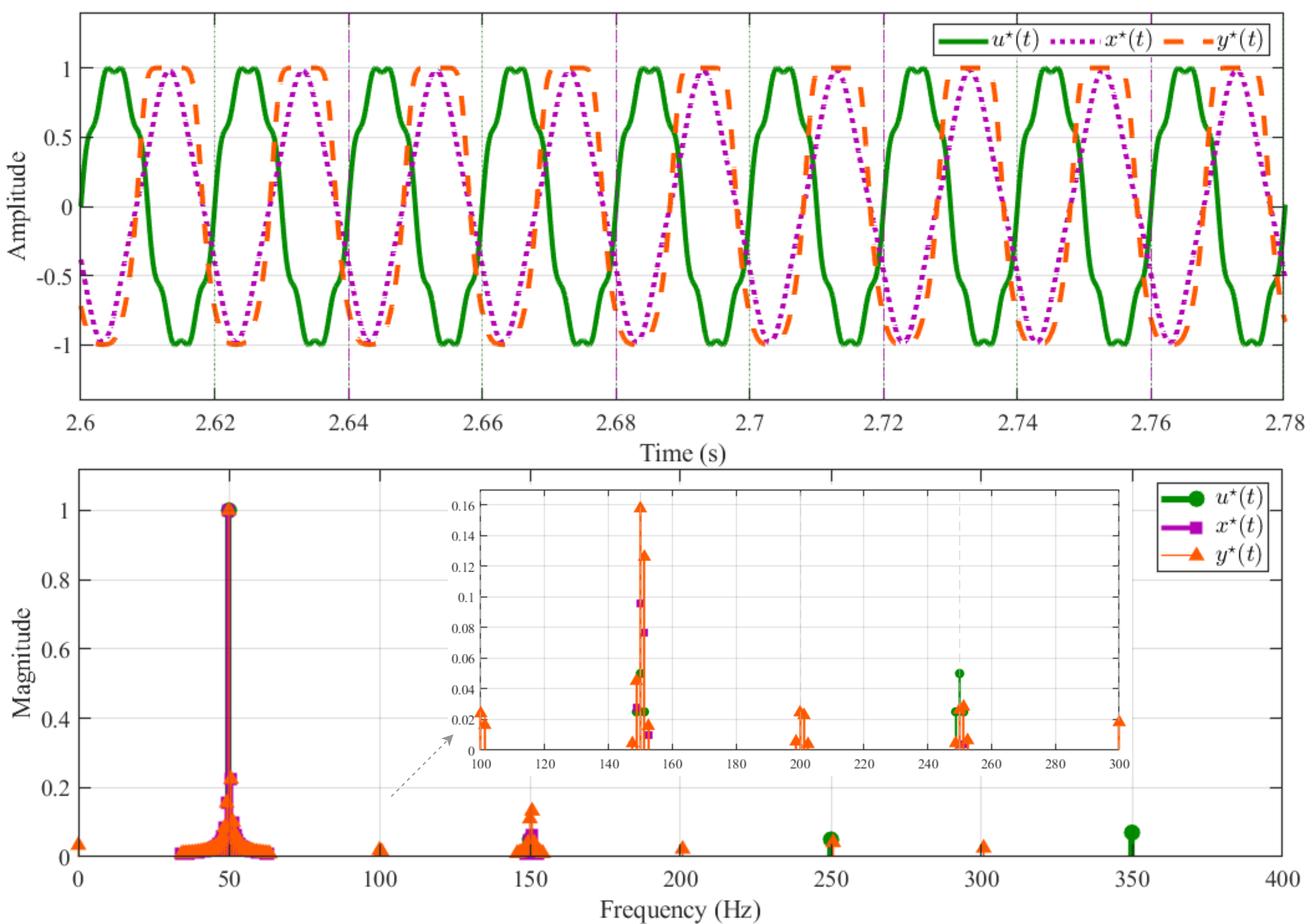}
  \caption{Illustrative nominal periodic motion $(u^\star(t),x^\star(t),y^\star(t))$ and corresponding spectra for a forced van der Pol--Duffing oscillator. While the input contains harmonics at 50, 150, 250, and 350~Hz, nonlinear dynamics amplify the 150 and 250~Hz components in the state response and generate additional components near 100, 200, and 300~Hz in the output response.}
    \label{fig:placeholder}
    \label{fig:nominal_periodic_motion}
\end{figure}
Also, a special case of a nominal periodic motion is a standard constant steady-state solution of~\eqref{eq:system} with a constant input $u^\star.$

We next introduce the following natural assumption.
\begin{assumption}\label{ass:periodic_trajectory}
The system~\eqref{eq:system} possesses a nominal periodic motion $(u^\star(t),x^\star(t),y^\star(t))$ and $T_1, T_2, T_3$ are commensurable, i.e., their ratios are rational numbers. 
\hfill $\square$
\end{assumption}

The commensurability assumption ensures the existence of a common period \(T=\mathrm{LCM}\{T_1, T_2, T_3\}\), where $\mathrm{LCM}$ denotes the least common multiple, which is required for the Fourier-series-based HSS representations employed later in the paper.

Consider the system~\eqref{eq:system} and a nominal periodic motion $(u^\star(t),x^\star(t),y^\star(t))$. Define the error coordinates \footnote{To simplify notation, the explicit time dependence of the signals $u$, $x$, $y$, and the error coordinates $\tilde{u}$, $\tilde{x}$, $\tilde{y}$ is omitted throughout the Section~\ref{sec:HS_formalization} and~\ref{sec:HS_assessment_tools}.}
\begin{align}\label{eq:error_cords}
\tilde u \coloneqq u-u^\star(t), ~
\tilde x \coloneqq x-x^\star(t), ~
\tilde y \coloneqq y-y^\star(t).
\end{align}
Then, the error dynamics of the system~\eqref{eq:system} evolve according to the following equations:
\begin{subequations}\label{eq:error_dynamics}
    \begin{align}
  \dot{\tilde x}
 & 
= f\bigl(x(t),u(t)\bigr)
- f\bigl(x^\star(t),u^\star(t)\bigr),\\
\tilde{y}& = h(x(t))-h(x^\star(t)).
\end{align}
\end{subequations}
Now, we provide the proposed definition of harmonic stability, which unifies several notions of harmonic stability employed in the literature \cite{Sun2011,Agorreta2011,Amin2017,Zhou2004HarmonicStateOperator,Kwon2017HarmonicSS,wereley1990analysis,becker2024}.

\begin{definition}[Harmonic Stability]\label{def:HS}
Consider the system~\eqref{eq:system} subject to Assumptions~\ref{assum:f} and~\ref{ass:periodic_trajectory}. Then, \eqref{eq:system} is \emph{harmonically stable} with respect to the nominal periodic motion $(u^\star(t),x^\star(t),y^\star(t))$ if the error dynamics \eqref{eq:error_dynamics} are small-signal BIBO stable, and if $\tilde{x} =\mathbf{0}_n$ is a locally  uniformly asymptotically  stable equilibrium of \eqref{eq:error_dynamics} for $\tilde{u}= \mathbf{0}_m$ $\forall t\geq 0$.
\hfill $\square$
\end{definition}

Definition~\ref{def:HS} provides a control-theoretic characterization of harmonic stability for nonlinear CBPSs operating around periodic trajectories. On the one hand, the BIBO stability requirement captures the bounded-response behavior of the system under sufficiently small perturbations of the nominal input, which is consistent with impedance- and Nyquist-based approaches reported in \cite{Sun2011,Agorreta2011,Amin2017}. On the other hand, the asymptotic stability requirement ensures that the nominal periodic motion constitutes a stable operating solution under the nominal input, which is related to harmonic-domain and HSS eigenvalue analyses reported in \cite{Zhou2004HarmonicStateOperator,Kwon2017HarmonicSS} as well as to LTP-based formulations \cite{wereley1990analysis,becker2024}. Consequently, the proposed definition of harmonic stability combines BIBO responses with convergence of the internal dynamics toward the nominal periodic motion. In that way, Definition~\ref{def:HS} consistently integrates the prevalent harmonic stability assessment results from the literature.

\vspace{-0.3cm}
\section{Harmonic Stability: Proposed Assessment Tools}\label{sec:HS_assessment_tools}

In this section, we propose two approaches to assess the harmonic stability of nonlinear CBPSs based on the corresponding LTP representation of the dynamics \eqref{eq:system} around a nominal periodic motion.
\vspace{-0.3cm}
\subsection{Harmonic Stability Assessment in Nonlinear CBPSs via Linearization}\label{sec:HS_nonlinear}

To characterize the local behavior of the error dynamics~\eqref{eq:error_dynamics} around the origin, we invoke the mean value theorem \cite{Rudin64} to write the following equation:
    \begin{align}\label{eq:mean_value_dirty}
        f(x,u) - f(x^\star(t),u^\star(t)) &=\frac{\partial f(x'(t),u'(t))}{\partial x}(x-x^\star(t)) \nonumber\\
        & \phantom{=}+\frac{\partial f(x'(t),u'(t))}{\partial u}(u-u^\star(t)),
    \end{align}
for some $(x'(t),u'(t))\in D_x\times D_u$ in the line segment connecting $(x,u)$ with  $(x^\star(t),u^\star(t))$. By adding and subtracting the terms $\frac{\partial f(x^\star(t),u^\star(t))}{\partial x}(x-x^\star(t))$ and $\frac{\partial f(x^\star(t),u^\star(t))}{\partial u}(u-u^\star(t))$ to/from the right-hand side of \eqref{eq:mean_value_dirty}, we get
\begin{subequations}\label{eq:mean_value_f}
\begin{align}
     f(x,u) - f(x^\star(t),u^\star(t)) = A(t)\tilde{x}+B(t)\tilde{u}+\phi(t,\tilde{x},\tilde{u}),
\end{align}
where
    \begin{align}
     A(t) & =\frac{\partial f(x^\star(t),u^\star(t))}{\partial x},~B(t)  = \frac{\partial f(x^\star(t),u^\star(t))}{\partial u},\\
\phi(t,\tilde{x},\tilde{u}) & =  \left(\frac{\partial f(x'(t),u'(t))}{\partial x}-\frac{\partial f(x^\star(t),u^\star(t))}{\partial x}\right)\tilde{x} \nonumber\\
& \phantom{=}+ \left(\frac{\partial f(x'(t),u'(t))}{\partial u}-\frac{\partial f(x^\star(t),u^\star(t))}{\partial u}\right)\tilde{u}.
    \end{align}
\end{subequations}
We can follow an analogous procedure to obtain that the output mapping $h$ satisfies the following equation:
\begin{subequations}\label{eq:mean_value_h}
\begin{equation}
    h(x)-h(x^\star(t))=C(t)\tilde{x}+\psi(t,\tilde{x})=:\tilde{y},
\end{equation}
where
\begin{equation}
    C(t)  = \frac{\partial h(x^\star(t))}{\partial x},~
    \psi(t,\tilde{x})  = \left(\frac{\partial h(x'(t))}{\partial x}-\frac{\partial h(x^\star(t))}{\partial x}\right)\tilde{x}.
\end{equation}
\end{subequations}
Then, with \eqref{eq:mean_value_f} and \eqref{eq:mean_value_h}, we obtain the following equivalent representation of the system \eqref{eq:error_dynamics} around the nominal periodic motion $(u^\star(t),x^\star(t),y^\star(t))$:
\begin{subequations}\label{eq:errorF_short}
    \begin{align}
\dot{\tilde x} & = A(t)\tilde{x} + B(t)\tilde{u} + \phi(t,\tilde{x},\tilde{u}),\\
\tilde{y} & = C(t)\tilde{x}+\psi(t,\tilde{x}). 
\end{align}
\end{subequations}
Since $f$ and $h$ satisfy Assumption~\ref{assum:f} and $(u^\star(t),x^\star(t),y^\star(t))$ is a periodic motion, $A(t)$, $B(t)$, and $C(t)$ are time-periodic matrices.

The following result provides a sufficient condition for local harmonic stability of the system~\eqref{eq:system} by using the representation~\eqref{eq:errorF_short}.
\begin{proposition}\label{thm:lin_to_local_iss}
  Consider the system~\eqref{eq:system} satisfying Assumptions \ref{assum:f}, 
  and \ref{ass:periodic_trajectory}. Then, the system \eqref{eq:system} is locally harmonically stable  with respect to the nominal periodic motion $(u^\star(t),x^\star(t), y^\star(t))$ in the sense of Definition~\ref{def:HS} if there exists a continuous, positive-definite symmetric matrix $P(t)$ satisfying $c_1I_{n} \preceq  P(t) \preceq  c_2I_{n}$ $\forall t\geq 0$, for some positive constants $c_1$ and $c_2$, as well as positive scalars $\gamma$ and $\epsilon$, such that  the following \emph{time-varying} matrix inequality is satisfied $\forall t\geq 0$:
\begin{equation}\label{eq:P(t)}
\resizebox{0.98\columnwidth}{!}{$
\begin{bmatrix}
\dot{P}(t)+A^\top(t)P(t)+P(t)A(t)+C^\top(t)C(t)+\epsilon I_n &
P(t)B(t)\\
B^\top(t)P(t) &
-\gamma^2 I_m
\end{bmatrix}
\preceq 0.
$}
\end{equation}
\hfill $\square$
\end{proposition}

\begin{IEEEproof}
See  Appendix~\ref{appendix_1}.
\end{IEEEproof}

\vspace{-0.3cm}
\subsection{LTP and HSS Representations}

A direct solution of the LMI~\eqref{eq:P(t)} is computationally challenging due to the time-periodic nature of the system matrices and the Lyapunov matrix function $P(t)$. To obtain a numerically more tractable formulation, we next employ an HSS representation for \eqref{eq:system} associated with the linearization of~\eqref{eq:errorF_short} around the nominal periodic motion $(u^\star(t),x^\star(t),y^\star(t))$. 
Thus, the approach pursued in the sequel is similar to those in \cite{wereley1990analysis,becker2024}, but with the important difference that we link stability properties of the HSS representation to the harmonic stability of the nonlinear CBPS \eqref{eq:system} in the sense of Definition~\ref{def:HS}.

For that purpose, due to Assumption~\ref{assum:f}, the nonlinear remainder terms in \eqref{eq:errorF_short} satisfy
\begin{subequations}
\begin{align}
\lim_{\|(\tilde x,\tilde u)\|\to0}
\sup_{t\ge0}
\frac{\|\phi(t,\tilde x,\tilde u)\|}
     {\|(\tilde x,\tilde u)\|}
= 0, \quad 
\lim_{\|\tilde x\|\to0}
\sup_{t\ge0}
\frac{\|\psi(t,\tilde x)\|}
     {\|\tilde x\|}
&= 0.
\end{align}
\end{subequations}
Hence, for sufficiently small $\|\tilde{x}\|$, $\|\tilde{u}\|$, we see from \eqref{eq:errorF_short} that the  dynamics~\eqref{eq:error_dynamics} admit the following LTP approximation:
\begin{subequations}\label{eq:linearized}
\begin{align}
\dot{\tilde x}(t)
&= A(t)\tilde x(t) + B(t)\tilde u(t), \\
\tilde y(t)
&= C(t)\tilde x(t).
\end{align}
\end{subequations}

Next, we follow the framework presented in~\cite{wereley1990analysis,DeRua2020}
to derive the exact HSS representation of the LTP system~\eqref{eq:linearized}. Recall that, by Assumptions~\ref{assum:f} and~\ref{ass:periodic_trajectory}, the matrices $A(t)$, $B(t)$, and $C(t)$ are continuous and $T$-periodic with common period $T:=\mathrm{LCM}(T_1,T_2,T_3)$, where $\mathrm{LCM}$ denotes the least common multiple~\cite{wereley1990analysis,liao2022}. Let $\omega_g:=2\pi/T$ denote the angular frequency associated with the common period $T$. Then, the periodic coefficient matrices in~\eqref{eq:linearized} admit the following Fourier-series representations~\cite{Kwon2017HarmonicSS,DeRua2020}:

\begin{subequations}\label{eq:Fourier_coefficients}
\small
\begin{align}
A(t) &= \lim_{N\to\infty}\sum_{k=-N}^{N} A_k e^{jk\omega_g t},
~
B(t) = \lim_{N\to\infty}\sum_{k=-N}^{N} B_k e^{jk\omega_g t}, \\
C(t) &= \lim_{N\to\infty}\sum_{k=-N}^{N} C_k e^{jk\omega_g t},
\end{align}
\end{subequations}

where $A_k$, $B_k$, and $C_k$ denote the corresponding complex harmonic coefficients of $A(t)$, $B(t)$, and $C(t)$, respectively.
Similarly, the state, input, and output trajectories of the LTP system~\eqref{eq:linearized} can be expressed as
\begin{subequations}\label{eq:sliding_sum}
\small
\begin{align}
\tilde x(t)
&=
\lim_{N\to\infty}
\sum_{k=-N}^{N}
\tilde X_k(t)e^{jk\omega_g t},
~
\tilde u(t)
=
\lim_{N\to\infty}
\sum_{k=-N}^{N}
\tilde U_k(t)e^{jk\omega_g t},
\\
\tilde y(t)
&=
\lim_{N\to\infty}
\sum_{k=-N}^{N}
\tilde Y_k(t)e^{jk\omega_g t},
\end{align}
\end{subequations}
where we note that the corresponding complex harmonic coefficients are time-dependent.

By considering \eqref{eq:Fourier_coefficients} and \eqref{eq:sliding_sum}, the LTP system~\eqref{eq:linearized} can be equivalently represented by the following LTI dynamics in HSS form~\cite{wereley1990analysis,mollerstedt2000,DeRua2020}:
\begin{equation}\label{eq:HSS_LTI_infinite_complex}
\begin{aligned}
\dot{\tilde{\mathcal X}}(t)
&=
(\mathcal A-\mathcal N)\tilde{\mathcal X}(t)
+
\mathcal B\,\tilde{\mathcal U}(t), \\
\tilde{\mathcal Y}(t)
&=
\mathcal C\,\tilde{\mathcal X}(t),
\end{aligned}
\end{equation}
where

{\small
\begin{equation}\label{eq:HSS_vectors}
\tilde{\mathcal X}(t)
:=
\lim_{N\to\infty}
\big(
\tilde X_{-N}(t),
\ldots,
\tilde X_{-1}(t),
\tilde X_{0}(t),
\tilde X_{1}(t),
\ldots,
\tilde X_{N}(t)
\big)^\top
\end{equation}}

with
\begin{equation}
\mathcal A
=
\lim_{N\to\infty}
\begin{bmatrix}
A_0      & A_{-1}   & \cdots & A_{-2N} \\
A_1      & A_0      & \ddots & \vdots \\
\vdots   & \ddots   & \ddots & A_{-1} \\
A_{2N}   & \cdots   & A_1    & A_0
\end{bmatrix},
\end{equation}
i.e. $\mathcal A$ has an infinite-dimensional block-Toeplitz structure,  and $\mathcal{B}$ and $\mathcal{C}$ are defined analogously to $\mathcal{A}$.
Moreover, 
\begin{equation}
\mathcal N
=
\lim_{N\to\infty}
\mathrm{blkdiag}
\left(
[-jN\omega_g I_n,\ldots,\mathbf 0_{n\times n},\ldots,jN\omega_g I_n]
\right)
\end{equation}
and the vectors $\tilde {\mathcal U}(t)$ and $\tilde {\mathcal Y}(t)$ are defined analogously to $\tilde {\mathcal X}(t)$.

The HSS realization~\eqref{eq:HSS_LTI_infinite_complex} provides an  LTI representation of the LTP dynamics~\eqref{eq:linearized}. Such a representation, as well as its associated harmonic transfer function (HTF)~\cite{wereley_GNC1990,Mollerstedt2000out_of_control,liao2022}, 
\begin{equation}\label{eq:HTF}
\mathcal G(s)
=
\mathcal C
\bigl(
sI-\mathcal A+\mathcal N
\bigr)^{-1}
\mathcal B,
\end{equation}
form the basis of several stability analysis methods, including Nyquist-, eigenvalue-, and Lyapunov-based approaches~\cite{wereley1990analysis,Zhou2004HarmonicStateOperator,salis2018,wang2019,DeRua2020}. 

\vspace{-0.2cm}
\subsection{A Tractable Approach to Harmonic Stability Assessment\label{sec:HS_practical_approach_LTP}}

In general, the  HSS (or the HTF) representation~\eqref{eq:HSS_LTI_infinite_complex} of an LTP system is infinite-dimensional\cite{wereley1990analysis, Bittanti1991}. 
In practice, however, finite-dimensional approximations are typically employed to obtain models suitable for numerical computations, stability analysis, and control design of CBPSs~\cite{mollerstedt2000,Salis2017,DeRua2020,yang2021,liao2022}.  Such finite-dimensional approximations of the HSS dynamics are associated with truncated representations of $(A(t),B(t),C(t))$ in \eqref{eq:Fourier_coefficients}  and  $(\tilde x(t),\tilde{u}(t),\tilde{y}(t))$ in \eqref{eq:sliding_sum},  obtained by considering only a finite set of harmonics. 

In this spirit, let $N\in \mathbb{N}$ be finite. Then, we obtain the following finite-dimensional HSS dynamics from ~\eqref{eq:HSS_LTI_infinite_complex}:
\begin{subequations}\label{eq:HSS_finite_complex}
\begin{align}
    \dot{\bm{\tilde X}}(t)
& =
(\bm{\mathcal A} - \bm{\mathcal N})\bm{\tilde X}(t)
+
\bm{\mathcal B}\,\bm{\tilde U}(t),\\
\bm{\tilde Y}(t)
&=
\bm{\mathcal C}\,\bm{\tilde X}(t),
\end{align}
\end{subequations}
where $(\bm{\mathcal A},\bm{\mathcal{B}},\bm{\mathcal{C}},\bm{\mathcal{N}})$ and $(\bm{\tilde X,\bm{\tilde{U},\bm{\tilde{Y}}}})$ are finite-dimensional analogous quantities of $(\mathcal{A},\mathcal{B},\mathcal{C},\mathcal{N})$ and $(\tilde{\mathcal X},\tilde{\mathcal U},\tilde{\mathcal Y})$ in \eqref{eq:HSS_LTI_infinite_complex}, respectively. 
To continue with our analysis,  we note that the HSS realization \eqref{eq:HSS_finite_complex} is complex-valued due to the complex Fourier coefficients and the frequency-shift operator $\bm{\mathcal N}$. Therefore, in contrast to \cite{wereley1990analysis,Zhou2004HarmonicStateOperator,salis2018,wang2019,DeRua2020}, we at first formulate~\eqref{eq:HSS_finite_complex} equivalently in the real domain, so that our subsequent analysis can rely on standard control-systems methods (cf.~\cite{Doria2026}). Let
\begin{subequations}\label{eq:HSS_real_img_state_decomposition}
\begin{align}
    \bm{\tilde X}(t)
&=
\bm{\tilde X}_{\mathrm r}(t)
+
j\bm{\tilde X}_{\mathrm i}(t),
~
\bm{\tilde U}(t)
&=
\bm{\tilde U}_{\mathrm r}(t)
+
j\bm{\tilde U}_{\mathrm i}(t),
\\
\bm{\tilde Y}(t)
&=
\bm{\tilde Y}_{\mathrm r}(t)
+
j\bm{\tilde Y}_{\mathrm i}(t),
\end{align}
\end{subequations}
and denote
$A_{\mathrm r}
:=
\Re(\bm{\mathcal A} - \bm{\mathcal N})$, 
$A_{\mathrm i}
:=
\Im(\bm{\mathcal A} - \bm{\mathcal N})$,  $B_{\mathrm r}
:=
\Re(\bm{\mathcal B})$, $B_{\mathrm i}
:=
\Im(\bm{\mathcal B})$, $C_{\mathrm r}
:=
\Re(\bm{\mathcal C})$ and $C_{\mathrm i}
:=
\Im(\bm{\mathcal C})$, so that 
\begin{equation}\label{eq:HSS_real_img_matrices_decomposition}
\bm{\mathcal A} - \bm{\mathcal N}
=
A_{\mathrm r}
+
jA_{\mathrm i},
\quad
\bm{\mathcal B}
=
B_{\mathrm r}
+
jB_{\mathrm i},
\quad
\bm{\mathcal C}
=
C_{\mathrm r}
+
jC_{\mathrm i}.
\end{equation}
Substituting the above decompositions into~\eqref{eq:HSS_finite_complex} and separating real and imaginary parts yields the following equivalent real-valued LTI system:
\begin{subequations}\label{eq:HSS_real_stacked}
\begin{equation}
\begin{aligned}
\dot{\bar{\bm X}}(t) & 
=
\bar{\mathcal A}\bar{\bm X}(t)
+
\bar{\mathcal B}\bar{\bm U}(t),\\
\bar{\bm Y}(t)
 & =
\bar{\mathcal C}\bar{\bm X}(t),
\end{aligned}
\end{equation}
where
\begin{equation}
\bar{\bm X}(t)
:=
\begin{bmatrix}
\bm{\tilde X}_{\mathrm r}(t)\\
\bm{\tilde X}_{\mathrm i}(t)
\end{bmatrix},
~
\bar{\bm U}(t)
:=
\begin{bmatrix}
\bm{\tilde U}_{\mathrm r}(t)\\
\bm{\tilde U}_{\mathrm i}(t)
\end{bmatrix},
~
\bar{\bm Y}(t)
:=
\begin{bmatrix}
\bm{\tilde Y}_{\mathrm r}(t)\\
\bm{\tilde Y}_{\mathrm i}(t)
\end{bmatrix},
\end{equation}
and 
\begin{equation}\label{eq:HSS_real_matrices}
\bar{\mathcal A}
=
\begin{bmatrix}
A_{\mathrm r} & -A_{\mathrm i}\\
A_{\mathrm i} & A_{\mathrm r}
\end{bmatrix},
~
\bar{\mathcal B}
=
\begin{bmatrix}
B_{\mathrm r} & -B_{\mathrm i}\\
B_{\mathrm i} & B_{\mathrm r}
\end{bmatrix},
~
\bar{\mathcal C}
=
\begin{bmatrix}
C_{\mathrm r} & -C_{\mathrm i}\\
C_{\mathrm i} & C_{\mathrm r}
\end{bmatrix}.
\end{equation}
\end{subequations}

 We are in a position to state our main result, which establishes sufficient conditions for assessing
 harmonic stability of the nonlinear CBPS \eqref{eq:system} with respect to the nominal periodic motion $(u^\star(t),x^\star(t),y^\star(t))$ via the solution of finite-dimensional matrix inequalities.

\begin{assumption}
\label{ass:the1}
Consider the system~~\eqref{eq:HSS_real_stacked}.
   There exist scalars $N_0 \in \mathbb{N}$, $\epsilon > 0$ and $\gamma>0$, such that for any $N \ge N_0$ there exist block-Toeplitz matrices $P_r = P_r^\top \succ 0$ and $P_i = -P_i^\top$, depending on $N$, such that the following \emph{time-invariant} matrix inequalities hold:
\begin{subequations}
\label{eq:HSS_BR_LMI}
\begin{align}
\begin{bmatrix}
\bar{\mathcal A}^\top \bar{\mathcal P}
+
\bar{\mathcal P}\bar{\mathcal A}
+
\bar{\mathcal C}^\top \bar{\mathcal C}   + \epsilon I_{2N}
&
\bar{\mathcal P} \bar{\mathcal B}
\\
\bar{\mathcal B}^\top \bar{\mathcal P}
&
-\gamma^2 I_{2N}
\end{bmatrix}
 & \preceq 0,\\
 \bar \cP = \begin{bmatrix}
     P_r & -P_i\\P_i & P_r
 \end{bmatrix} &  \succ 0.
\end{align}
\end{subequations} 
\hfill $\square$
\end{assumption}
\begin{theorem}\label{thm:practical_harmonic_stab}
 Consider the system~\eqref{eq:system} subject to Assumptions~\ref{assum:f} and \ref{ass:periodic_trajectory} as well as its corresponding HSS error dynamics~\eqref{eq:HSS_real_stacked} defined with respect to the nominal periodic motion $(u^\star(t),x^\star(t),y^\star(t))$. 
Then, with Assumption~\ref{ass:the1}, the  system~\eqref{eq:system} is harmonically stable with respect to $(u^\star(t),x^\star(t),y^\star(t))$ in the sense of Definition~\ref{def:HS}.
\hfill $\square$
\end{theorem}
\begin{IEEEproof}
See Appendix~\ref{appendix_2}.
\end{IEEEproof}

Often, in power systems applications, the relevant set of harmonics might be finite. Then, the resulting HSS representation is also directly finite dimensional, and the stability condition in Assumption~\ref{ass:the1} has only to be verified for a single, fixed $N_0=N$.
This leads to the following corollary to Theorem~\ref{thm:practical_harmonic_stab}
for the presentation of which, we introduce the sets 
\begin{equation}
  \begin{split}
     \cU_{N_2} &:= \{\tilde u(t) \mid \tilde u(t) = \sum_{k=-N_2}^{N_2}\tilde U_k(t) e^{j k \omega_g t}\},\\
     \cX_{N_3} &: = \{\tilde x(t) \mid \tilde x(t) = \sum_{k=-N_3}^{N_3}\tilde X_k(t) e^{j k \omega_g t}\}, 
  \end{split} 
  \notag
\end{equation}
with $N_2\in \mathbb{N}$ and $N_3\in \mathbb{N}$, together with the assumption below.
\begin{assumption}
\label{ass:cor}
Consider the system \eqref{eq:system} under Assumptions \ref{assum:f} and  \ref{ass:periodic_trajectory}.
The nominal periodic motion $(u^\star(t),x^\star(t),y^\star(t))$ has a finite set of harmonics of order $N_1\in \mathbb{N}$. Moreover, for all  $\tilde u(t) \in \cU_{N_2}$, there exists an $N_3$, such that $\tilde x(t) \in \cX_{N_3}$.
\hfill $\square$
\end{assumption}
\smallskip
\begin{corollary}\label{corollary:1}
    Consider the system \eqref{eq:system} under Assumptions \ref{assum:f},~\ref{ass:periodic_trajectory} and \ref{ass:cor} as well as its corresponding HSS error dynamics~\eqref{eq:HSS_real_stacked} defined with respect to the nominal periodic motion $(u^\star(t),x^\star(t),y^\star(t))$. 
   Suppose that Assumption~\ref{ass:the1} holds for $N_0=N = \max\{N_1,N_2,N_3\}$. Then, the system~\eqref{eq:system} is harmonically stable with respect to $(u^\star(t),x^\star(t),y^\star(t))$ in the sense of Definition~\ref{def:HS}.
\hfill $\square$
\end{corollary}
\begin{IEEEproof}
The proof is straightforward by noting that under the corollary conditions, the finite-dimensional LTI approximation \eqref{eq:HSS_real_stacked} of \eqref{eq:linearized} is exact. 
\end{IEEEproof}

\subsection{Comparison with Respect to Existing Stability Characterizations}

The proposed harmonic stability assessment is directly related to existing HSS- and HTF-based stability methods~\cite{wereley1990analysis,mollerstedt2000,Zhou2004HarmonicStateOperator,salis2018,liao2022}. Indeed, the real-valued HSS realization~\eqref{eq:HSS_real_stacked} is equivalent to the conventional complex-valued realization~\eqref{eq:HSS_finite_complex}. Moreover, the corresponding state-matrix spectra satisfy~\cite{Doria2026} \begin{equation} \sigma(\bar{\mathcal{A}}) = \sigma(\bm{\mathcal{A}}-\bm{\mathcal{N}}) \cup \overline{\sigma(\bm{\mathcal{A}}-\bm{\mathcal{N}})}. \end{equation} Hence, $\bar{\mathcal A}$ is Hurwitz if and only if $\bm{\mathcal A}-\bm{\mathcal N}$ is Hurwitz, meaning that both realizations have the same internal stability properties. 
In addition, since the existence of a symmetric positive-definite matrix $\bar{\mathcal P}=\bar{\mathcal P}^{\top}\succ0$ and scalars $\epsilon>0$ and $\gamma>0$ satisfying~\eqref{eq:HSS_BR_LMI} is necessary and sufficient for $\bar{\mathcal A}$ to be Hurwitz it follows that under minimality of the HSS realization~\eqref{eq:HSS_real_stacked}, the generalized Nyquist criterion applied to the corresponding HTF is equivalent to the Hurwitz property of $\bm{\mathcal A}-\bm{\mathcal N}$~\cite{wereley_GNC1990,liao2022}. 

Thus, the proposed characterization through Theorem~\ref{thm:practical_harmonic_stab} (or Corollary~\ref{corollary:1}) is consistent with existing HSS- and HTF-based stability assessment methods, with the additional important feature that it rigorously links harmonic stability of the HSS representation~\eqref{eq:HSS_real_stacked} to the local harmonic stability of the nonlinear CBPS~\eqref{eq:system}.

\section{Application Example}\label{sec:case_study}
In this section, we demonstrate how local harmonic stability can be assessed through trajectory-based linearization combined with HSS analysis for an averaged nonlinear CBPS model of a three-phase grid-following voltage-source converter (VSC) system equipped with an LCL filter, a proportional-integral (PI) current controller implemented in the $dq$ frame, and a synchronous reference frame phase-locked loop (SRF-PLL).

\begin{figure*}[!t]
    \centering
    \includegraphics[width=\linewidth]{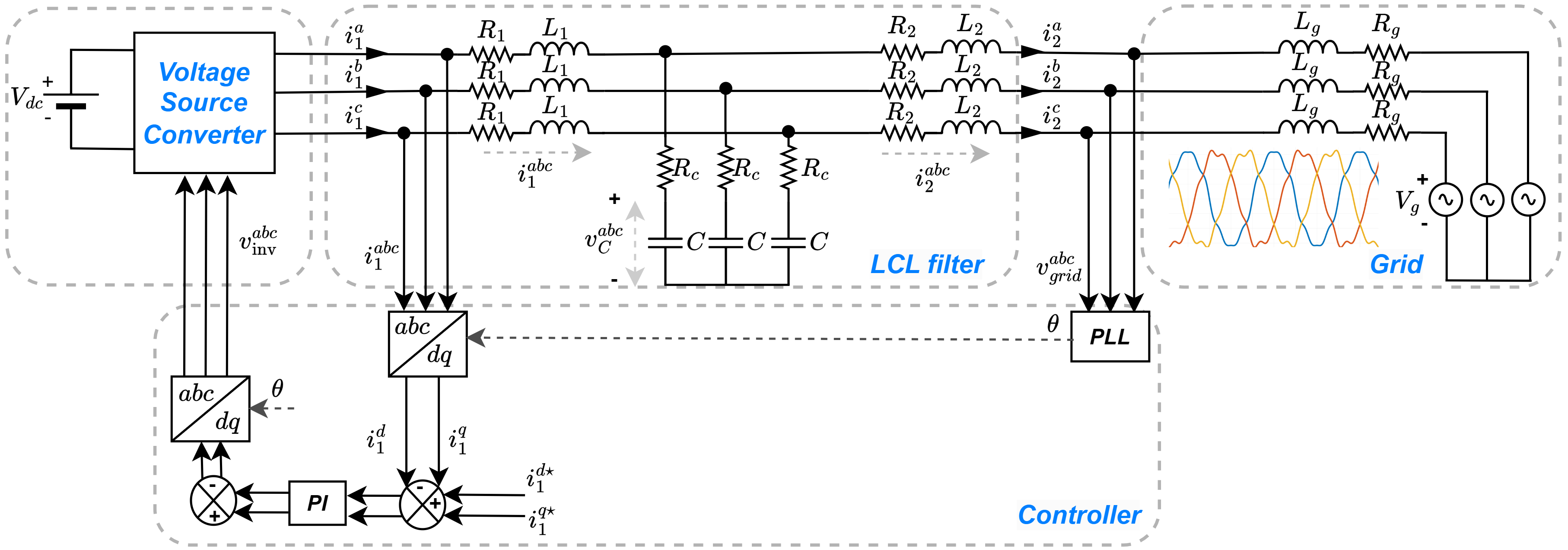}
    \caption{Averaged model of the considered grid-following converter, including the LCL filter, current controller, and grid.}
    \label{fig:GFL}
\end{figure*}
\vspace{-0.3cm}
\subsection{System Model}
In Fig.~\ref{fig:GFL}, we show the system setup considered. The dynamics are formulated in the stationary $abc$ reference frame and describe a balanced three-phase three-wire system. The electrical dynamics are given by
\begin{subequations}\label{eq:plant_model_EXAMPLE}
{\small
\begin{align}
L_1 \dot{i}_1^{abc} &=
v_{\mathrm{inv}}^{abc}
- v_C^{abc}
- R_1 i_1^{abc}
- R_c \bigl(i_1^{abc} - i_2^{abc}\bigr),
\\
C_f \dot{v}_C^{abc} &=
i_1^{abc} - i_2^{abc},
\\
(L_2 + L_g)\dot{i}_2^{abc}
&=
v_C^{abc}
+ R_c \bigl(i_1^{abc} - i_2^{abc}\bigr)
- v_{\mathrm{grid}}^{abc}
\nonumber\\
&\quad
- (R_2 + R_g)i_2^{abc},
\end{align}
}
\end{subequations}
where 
$i_1^{abc}$,
$i_2^{abc}$,
$v_C^{abc}$,
and
$v_{\mathrm{inv}}^{abc}$
denote the converter-side filter current, grid-side filter current, capacitor voltage, and converter-side control input voltage, respectively, all expressed in stationary $abc$ coordinates. Moreover, the grid voltage $v_{\mathrm{grid}}^{abc}(t)=:u(t)$ is treated as an external input to the converter~dynamics.

For control design and SRF-PLL implementation, the converter variables are transformed from the stationary $abc$ frame to the synchronous $dq$ frame using the Park transformation~\cite{Teodorescu2011_PLL}.

The corresponding transformations are denoted by $T_{dq}^{abc}: abc \rightarrow dq$ and $T_{abc}^{dq}: dq \rightarrow abc$. Using 
$T_{dq}^{abc}$, the $dq$ components of the grid voltage, capacitor voltage, and converter-side filter current are defined as
\begin{subequations}
\begin{align}
v_{dq}^{\mathrm{grid}}
=
\begin{bmatrix}
v_d^{\mathrm{grid}} \\
v_q^{\mathrm{grid}}
\end{bmatrix}
&=
T_{dq}^{abc}\,
v_{\mathrm{grid}}^{abc},\quad 
v_{dq}
=
\begin{bmatrix}
v_d \\
v_q
\end{bmatrix}
=
T_{dq}^{abc}\,
v_C^{abc},
\\
i_{dq}
=
\begin{bmatrix}
i_d \\
i_q
\end{bmatrix}
&=
T_{dq}^{abc}\,
i_1^{abc}.
\end{align}
\end{subequations}

Following the SRF-PLL structure commonly used in grid-connected converters~\cite{Chung2000,Teodorescu2011_PLL}, the PLL input is chosen as the $q$-axis component of the grid voltage, i.e., $v_q := v_q^{\mathrm{grid}}$. Then, the SRF-PLL dynamics can be described by
\begin{align}\label{eq:SFR_PLL_dynamics}
\dot{\xi}_{\mathrm{PLL}}
&= v_q,\qquad
\omega_{\mathrm{PLL}}
= \omega_g
+ k_{p,\mathrm{PLL}}v_q
+ k_{i,\mathrm{PLL}}\xi_{\mathrm{PLL}}.
\end{align}

For our purposes, it is more convenient to represent the PLL phase dynamics on the unit circle rather than directly integrating the phase angle $\theta$, since this formulation preserves the invariant manifold $z_c^2+z_s^2=1$. Furthermore, the equivalent unit-circle dynamics,
\begin{equation}\label{eq:unit_circle}
\dot{z}_c = -\omega_{\mathrm{PLL}}\, z_s, \quad 
\dot{z}_s = \phantom{-}\omega_{\mathrm{PLL}}\, z_c,
\end{equation}
where $z_c = \cos\theta,~z_s = \sin\theta.$ Then, the PI controller for regulating the converter-side filter current $i_1^{abc}$ is implemented in synchronous $dq$ coordinates
as follows~\cite{Teodorescu2011_PLL}:
\begin{subequations}\label{eq:control_system}
\begin{align}
\dot{\xi}_d &= i_d^\star - i_d, \\
v_d^{\mathrm{ctrl}} &=
v_d + k_p(i_d^\star - i_d)
+ k_i \xi_d
- \omega_{\mathrm{PLL}} L_1 i_q,
\\
\dot{\xi}_q &= i_q^\star - i_q, \\
v_q^{\mathrm{ctrl}} &=
v_q + k_p(i_q^\star - i_q)
+ k_i \xi_q
+ \omega_{\mathrm{PLL}} L_1 i_d.
\end{align}
By defining the overall closed-loop state vector as
\[
x =
\begin{bmatrix}
i_1^{abc},\;
i_2^{abc},\;
v_C^{abc},\;
\xi_d,\;
\xi_q,\;
\xi_{\mathrm{PLL}},\;
z_c,\;
z_s
\end{bmatrix}^{\top}
\in \mathbb{R}^{14},
\]
the converter control input can be expressed through the \emph{nonlinear} control law:
\begin{equation}\label{eq:benchmark_control_law}
v_{\mathrm{inv}}^{abc}
=
\gamma(x)
=
T_{abc}^{dq} v_{dq}^{\mathrm{ctrl}}.
\end{equation}
\end{subequations}
\noindent
treating the grid voltage $v_\mathrm{grid}^\mathrm{abc}(t)$ as the external input $u$, and considering the plant dynamics \eqref{eq:plant_model_EXAMPLE}, the SRF-PLL dynamics \eqref{eq:SFR_PLL_dynamics}, and the control law \eqref{eq:control_system}, the resulting closed-loop converter dynamics can be represented in the form of~\eqref{eq:system}. In the sequel, the corresponding closed-loop converter system is denoted by $\Sigma_{\mathrm{conv.cl}}$.  
The system and control parameters considered in the case study are summarized in Table~\ref{tab:converter_parameters}. 
\begin{table}[!htb]
\centering
\caption{System and Control Parameters.}
\label{tab:converter_parameters}
\renewcommand{\arraystretch}{1.2}
\begin{tabular}{l l}
\hline
\textbf{Parameter} & \textbf{Value} \\
\hline
DC link voltage $V_{\mathrm{dc}}$
& $800~\mathrm{V}$ \\

Filter inductances $(L_1,L_2)$
& $(1~\mathrm{mH},1~\mathrm{mH})$ \\

Grid-side inductance $L_g$
& $2~\mathrm{mH}$ \\

Filter capacitance $C_f$
& $20~\mu\mathrm{F}$ \\

Inductor resistances $(R_1,R_2)$
& $(0.1~\Omega,0.1~\Omega)$ \\

Grid resistance $R_g$
& $2~\Omega$ \\

Active damping resistance $R_c$
& $1~\Omega$ \\

Current PI gains $(k_p,k_i)$
& $(3,2500)$ \\

PLL gains $(k_{p,\mathrm{PLL}},k_{i,\mathrm{PLL}})$
& $(1,50)$ \\

$d$-axis current reference $i_d^\star$
& $5~\mathrm{A}$ \\

$q$-axis current reference $i_q^\star$
& $0~\mathrm{A}$ \\

Grid voltage magnitude $V_g$
& $230\sqrt{2}~\mathrm{V}$ \\

Grid frequency $f$
& $50~\mathrm{Hz}$ \\

Angular frequency $\omega_g$
& $2\pi \cdot 50~\mathrm{rad/s}$ \\
\hline
\end{tabular}
\end{table}

\subsection{Trajectory-Based Linearization and HSS Construction}

Let $\mathcal{K}\subset \mathbb{N}$ be given. Then, we represent possible stationary grid voltage signals $u^\star(t)$  as follows:
{\small
\begin{equation}
\label{eq:u_star_harm}
u^\star(t)=V_g
\begin{bmatrix}
\sin(\omega_g t)\\
\sin(\omega_g t-\tfrac{2\pi}{3})\\
\sin(\omega_g t+\tfrac{2\pi}{3})
\end{bmatrix}
+\!\!\sum_{k\in\mathcal K}V_{g,k}
\begin{bmatrix}
\sin(k\omega_g t)\\
\sin(k\omega_g t-\sigma_k\tfrac{2\pi}{3})\\
\sin(k\omega_g t+\sigma_k\tfrac{2\pi}{3})
\end{bmatrix},
\end{equation}
}
\noindent
where $V_g,V_{g,k},\omega_g>0$. Note that $u^\star(t)$ denotes the nominal voltage conditions of the grid perturbed by the given harmonics, which will allow us to analyze the dependence of the harmonic stability on the operating trajectory. For numerical experiments, we specifically consider harmonics of 3rd, 5th, and 7th orders with respective amplitudes of $0.03$, $0.05$, and $0.05$.

\begin{figure*}[!t]
\centering
\begin{subfigure}[t]{0.49\textwidth}
    \centering
    \includegraphics[width=0.9\linewidth]{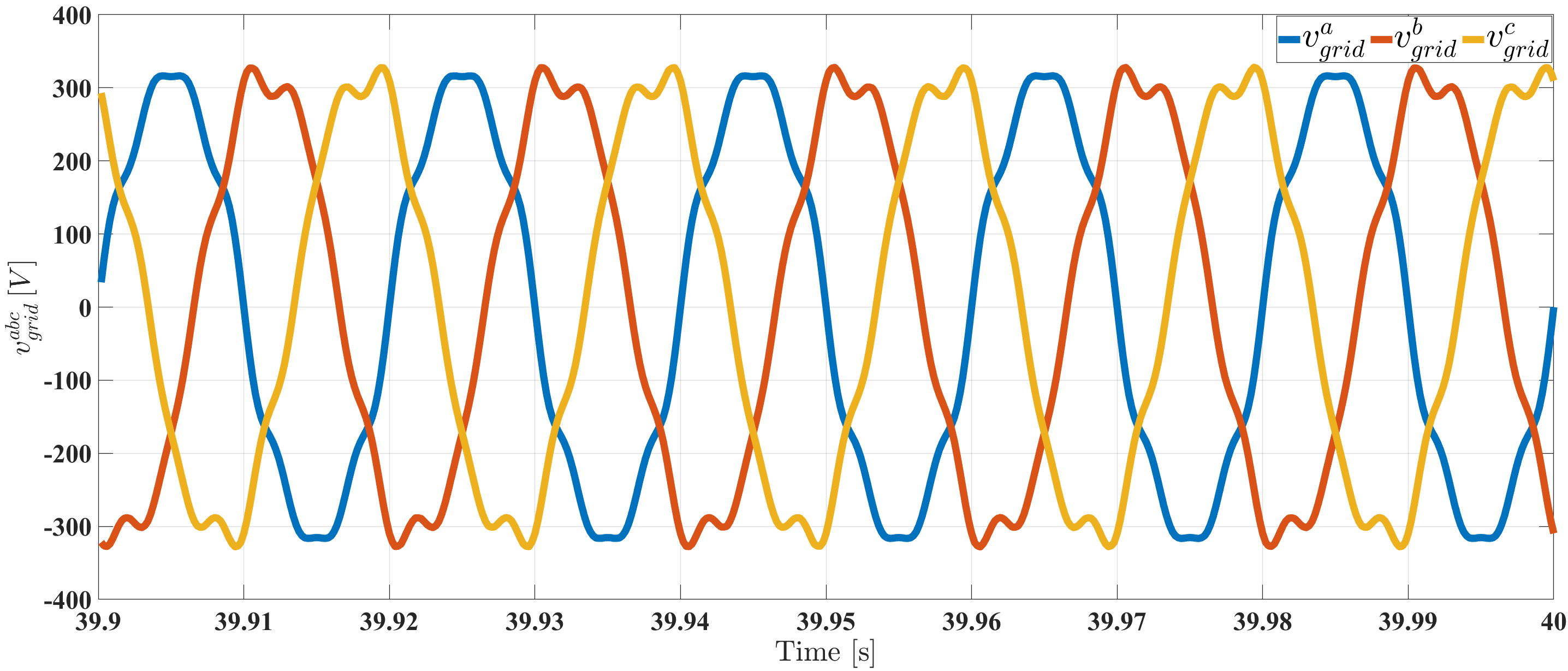}
    \caption{Grid-voltage $v^{abc}_{grid}(t)$ under harmonic excitation.}
    \label{fig:grid_voltage_waveforms}
\end{subfigure}
\hfill
\begin{subfigure}[t]{0.49\textwidth}
    \centering
    \includegraphics[width=0.9\linewidth]{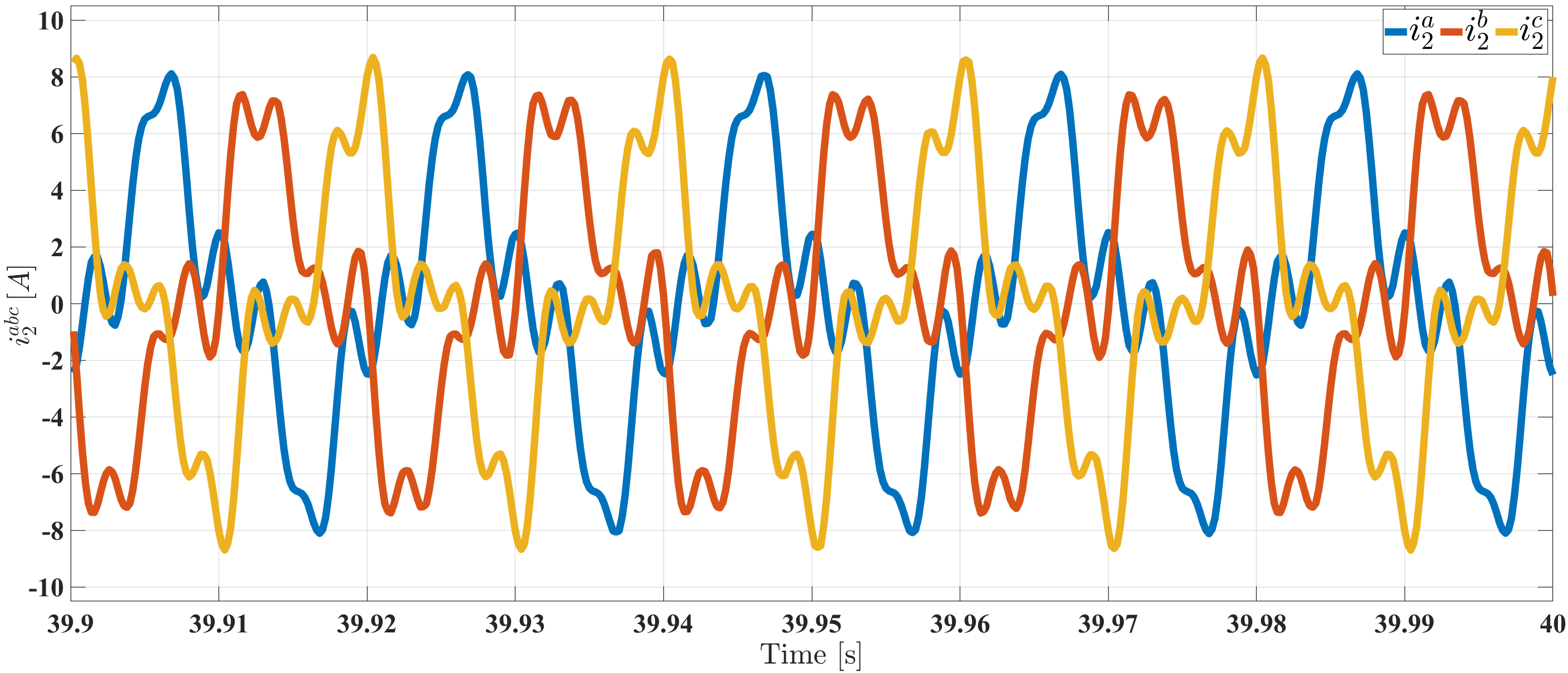}
    \caption{Grid-side filter current $i^{abc}_{2}(t)$ under harmonic excitation.}
    \label{fig:converter_current_waveforms}
\end{subfigure}

\vspace{0.1cm}

\begin{subfigure}[t]{0.49\textwidth}
    \centering
    \includegraphics[
        width=\linewidth,
        height=3.7cm,
        keepaspectratio
    ]{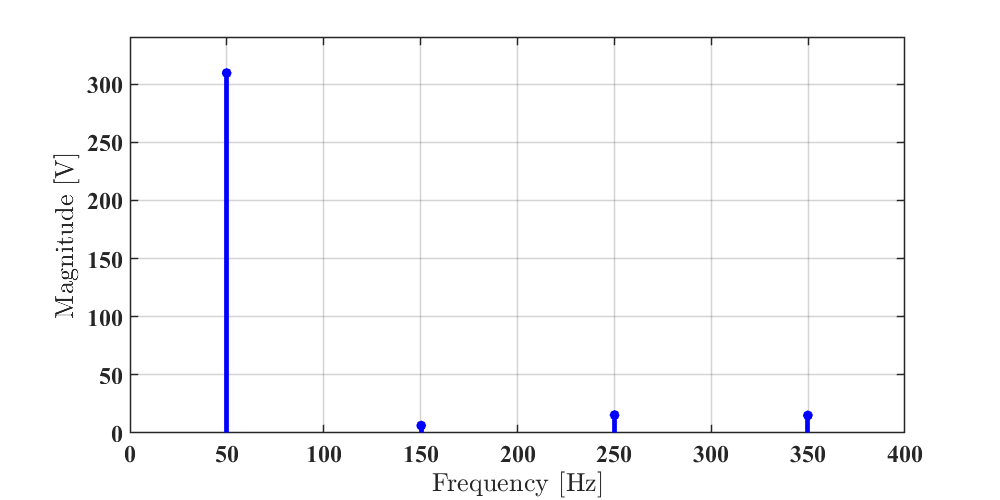}
    \caption{FFT spectrum of $v^{abc}_{grid}(t)$.}
    \label{fig:grid_voltage_fft}
\end{subfigure}
\hfill
\begin{subfigure}[t]{0.49\textwidth}
    \centering
    \includegraphics[
        width=\linewidth,
        height=3.9cm,
        keepaspectratio
    ]
    {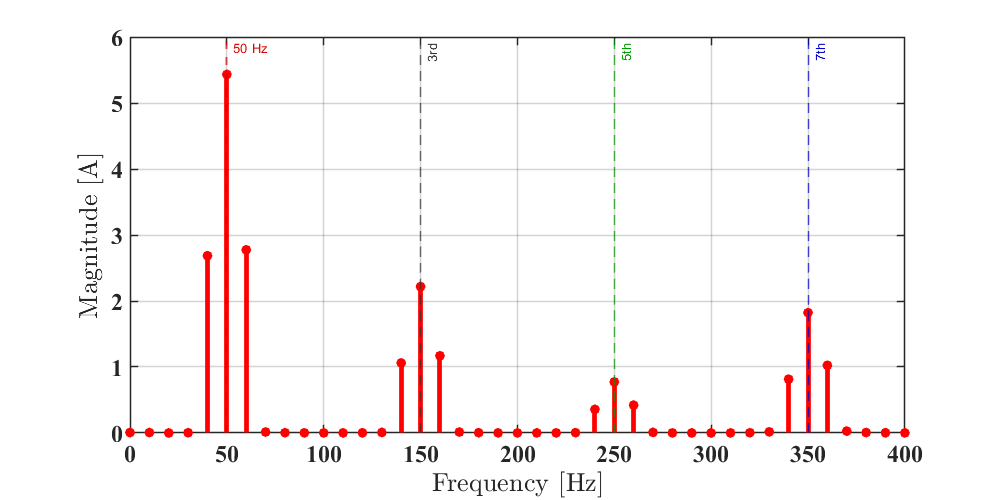}
    \caption{FFT spectrum of $i^{abc}_{2}(t)$.}
    \label{fig:converter_current_fft}
\end{subfigure}
\vspace{-0.1cm}
\caption{Three-phase grid voltage $v_{\mathrm{grid}}^{abc}(t)$ and grid-side filter current $i_2^{abc}(t)$ waveforms and corresponding FFT spectra under harmonic grid excitation.}
\label{fig:harmonic_analysis}
\end{figure*}

Taking into account the grid-voltage condition $u^\star(t)$ as in~\eqref{eq:u_star_harm}, the converter's closed-loop dynamics $\Sigma_\mathrm{conv.cl}$ admits a periodic steady-state trajectory $x^\star(t)$. Fig.~\ref{fig:harmonic_analysis} presents the resulting steady-state time-domain waveforms and corresponding FFT spectra. Note that the grid-voltage waveforms, shown in Fig.~\ref{fig:grid_voltage_waveforms}, exhibit visible distortion due to the presence of higher-order harmonic components. Moreover, the corresponding grid-side filter current, shown in Fig.~\ref{fig:converter_current_waveforms}, 
becomes significantly distorted under harmonic excitation. Furthermore, the current spectrum exhibits additional spectral components, revealing frequency-coupling effects in $\Sigma_{\mathrm{conv.cl}}$.

Having characterized $u^\star(t)$ and $x^\star(t)$, we can follow the trajectory-based linearization procedure introduced in Subsection~\ref{sec:HS_practical_approach_LTP}, to produce an LTP approximation of $\Sigma_\mathrm{conv.cl}$ along  $(x^\star(t),u^\star(t))$. Note, however, that due to the representation of the PLL phase through the unit-circle coordinates $(z_c,z_s)$ satisfying $z_c^2 + z_s^2 = 1$ (see \eqref{eq:unit_circle}), the periodic trajectory evolves on a constrained manifold. To obtain a linearization consistent with the manifold constraint, the Jacobian is projected onto the tangent space of the manifold. Let
\begin{equation}
z =
\begin{bmatrix}
z_c \\
z_s
\end{bmatrix},
\qquad
P_z = I_2 - zz^\top,
\end{equation}
and define the full-state projection matrix $P = \mathrm{diag}(I, P_z)$. The matrix $A(t)$ is then replaced by its projected counterpart $A(t) \leftarrow P\,A(t)\,P,$ which removes components orthogonal to the unit-circle manifold while preserving the local dynamics. Moreover, due to periodic excitation, we have $A(t+T) = A(t)$ and $B(t+T) = B(t)$. Although the internal phase $\theta$ evolves continuously, all state dependencies are periodic through the $(z_c,z_s)$ representation, ensuring overall time-periodicity of the linearized system.

Following the procedure described in Section~\ref{sec:HS_practical_approach_LTP}, the projected LTP model is transformed into a finite-dimensional HSS realization of the form \eqref{eq:HSS_finite_complex}. Subsequently, the equivalent real-valued representation \eqref{eq:HSS_real_stacked} with matrices $(\bar {\mathcal{A}},\bar  {\mathcal{B}},\bar  {\mathcal{C}})$ is constructed. For brevity, these matrices are not reported explicitly. The required Fourier coefficients were computed numerically from the periodic trajectories via FFT-based harmonic analysis over the finite harmonic set $\mathcal H=\{-N,\ldots, N\}$, where \(N \in \mathbb{N}\) denotes the harmonic truncation order. 

\vspace{-0.4cm}
\subsection{Results and Analysis}
Using the real-valued HSS realization~\eqref{eq:HSS_real_stacked} with matrices $(\bar{\mathcal A},\bar{\mathcal B},\bar{\mathcal C})$, the time-invariant matrix inequality~\eqref{eq:HSS_BR_LMI} is solved numerically in MATLAB using YALMIP and MOSEK. For each harmonic truncation order $N$, the matrix $\bar{\mathcal P}$ and the $\gamma$ are computed for a prescribed strictness margin $\epsilon>0$. 
The simulation and HSS parameters used throughout the case study are summarized in Table~\ref{tab:simulation_parameters}.
\begin{table}[H]
\centering
\caption{Simulation and HSS Settings.}
\label{tab:simulation_parameters}
\renewcommand{\arraystretch}{1.2}
\begin{tabular}{l l}
\hline
\textbf{Parameter} & \textbf{Value} \\
\hline
State dimension $n_x$
& $14$ \\

Input dimension $n_u$
& $3$ \\

HSS truncation order $N$
& $7$ \\

Simulation horizon
& $t \in [0,40]~\mathrm{s}$ \\
\hline
\end{tabular}
\end{table}

Table~\ref{tab:hss_metrics} summarizes the HSS stability metrics obtained for different harmonic truncation orders. The grid voltage contains third-, fifth-, and seventh-order harmonic components with amplitudes of $0.03$, $0.05$, and $0.05$, respectively.

\begin{table}[!htb] \centering \caption{HSS stability metrics for different truncation orders.} \renewcommand{\arraystretch}{1.1} \begin{tabular}{c c c} \hline Truncation Order $N$ & $\max \Re(\lambda)$ & $\gamma$ \\ \hline 1 & $-1.4209\times10^{-5}$ & 0.4545 \\ 3 & $-1.5284\times10^{-5}$ & 0.4545 \\ 5 & $-1.5600\times10^{-5}$ & 0.4545 \\ 7 & $-1.5750\times10^{-5}$ & 0.4545 \\ 9 & $-1.5836\times10^{-5}$ & 0.4545 \\ 11 & $-1.5893\times10^{-5}$ & 0.4545 \\ \hline \end{tabular} \label{tab:hss_metrics} \end{table}
The results show that the dominant eigenvalues of the HSS state matrix $\bar{\mathcal A}$ remain in the open left-half plane for all the truncation orders considered. Furthermore, as the truncation order $N$ increases, the eigenvalues of the HSS realization~\eqref{eq:HSS_real_stacked} converge to nearly invariant limiting values,  
while the $\gamma$ remains unchanged. The LMI is feasible for the prescribed strictness margin $\epsilon=2.6146\times10^{-5}$ for all truncation orders considered.  
These results indicate that the conditions of Theorem~\ref{thm:practical_harmonic_stab} are satisfied for all truncation orders considered in this case study, thereby establishing the local harmonic stability of the considered nonlinear CBPSs.
\vspace{-0.2cm}
\section{Conclusions}\label{sec:conclusions}

In this paper, we have introduced a control-theoretic notion of harmonic stability for nonlinear CBPSs. Concretely, we formulated harmonic stability as a combination of small-signal BIBO stability and local internal stability with respect to a nominal periodic trajectory, thereby unifying harmonic-domain interpretations with classical input--output stability concepts.  Moreover, we have shown that the local harmonic stability of nonlinear CBPSs can be analyzed through the properties of their LTP approximations around periodic operating trajectories. In particular, under suitable smoothness assumptions on the nonlinear dynamics, local harmonic stability can be certified through time-varying matrix inequalities associated with the system's LTP approximation. Based on this framework, we proposed a computationally tractable approach for assessing harmonic stability through HSS representations together with finite-dimensional LMIs.

The proposed framework was demonstrated on a grid-following converter under harmonically distorted grid conditions. The results highlighted the influence of harmonic operating conditions and model truncation on stability margins and demonstrated the importance of capturing cross-frequency coupling effects in CBPSs.
Future work will focus on harmonic-aware control design for nonlinear CBPSs, including extensions to interconnected converter-dominated networks, control strategies to improve harmonic stability margins and reduce harmonic distortion under distorted operating conditions, and formal truncation guarantees for harmonic state-space models.
\appendices


\section{Proof of Proposition \ref{thm:lin_to_local_iss}}\label{appendix_1}

    Under the conditions of Proposition~\ref{thm:lin_to_local_iss}, suppose that there exists a continuously differentiable, positive-definite matrix function $t\mapsto P(t) \succ 0$ and positive constants $\epsilon$ and $\gamma$ such that
\eqref{eq:P(t)} is true for all $t \geq 0$.
Define a positive definite function $(t,\tilde{x})\mapsto V(t, \tilde x) = \tilde x^\top P(t)\,\tilde x$ with $P(t)$ satisfying~\eqref{eq:P(t)}. From \cite[Theorem~4.12]{khalil}, we have that $c_1 I_n \preceq P(t) \preceq c_2I_n, \, c_1,c_2 >0$ for all $t \geq 0$, implies that 
\begin{equation}\label{eq:bounds_on_V}
\alpha_1(\Vert \tilde{x}\Vert_2):=c_1\Vert\tilde{x}\Vert_2^2 \leq V(t,\tilde{x}) \leq c_2\Vert\tilde{x}\Vert_2^2=:\alpha_2(\Vert \tilde{x}\Vert_2).
\end{equation}
Evaluating the time-derivative  of $V$ along the solutions of \eqref{eq:errorF_short} we obtain the following:
\begin{align}\label{eq:dot_V_1}
  \dot V
  &= \tilde x^T\bigl(\dot P(t) + A^\top(t) P(t) + P(t)\,A(t)\bigr)\,\tilde x  \nonumber\\
  &\phantom{=}+\tilde x^\top (P(t)\,B(t) + B^\top(t)P(t))\,\tilde u \nonumber\\
  & \phantom{=} + 2\tilde{x}^\top P(t)\phi(t,\tilde{x},\tilde{u}) \nonumber\\
  & = \mathcal{Q}(t,\tilde{x},\tilde{u}) +2\tilde{x}^\top P(t)\phi(t,\tilde{x},\tilde{u}),
\end{align}
where
\begin{align*}
&\mathcal{Q}(t,\tilde{x},\tilde{u})\\
&=\begin{bmatrix}
      \tilde{x}\\
      \tilde{u}\end{bmatrix}^\top \begin{bmatrix}
          \dot{P}(t)+A^\top(t)P(t)+P(t)A(t) & P(t)B(t)\\
          B^\top(t)P(t) & 0
      \end{bmatrix}\begin{bmatrix}
          \tilde{x}\\
          \tilde{u}
      \end{bmatrix}.
\end{align*}
We proceed to compute a suitable upper bound for $\dot V$ in \eqref{eq:dot_V_1}. Let us first note that since  the matrix inequality \eqref{eq:P(t)} holds, the following holds for all $(t\geq0,\tilde{x},\tilde{u})$:
\begin{align}\label{eq:bound_calQ_1}
\mathcal{Q}(t,\tilde{x},\tilde{u}) & \leq \begin{bmatrix}
      \tilde{x}\\
      \tilde{u}\end{bmatrix}^\top \begin{bmatrix}
         -C^\top(t) C(t)-\epsilon I_n & 0\\
          0 & \gamma^2 I_m
      \end{bmatrix}\begin{bmatrix}
          \tilde{x}\\
          \tilde{u}
      \end{bmatrix} \nonumber\\
      & = -\tilde{x}^\top C^\top (t)C(t) \tilde{x}-\epsilon \tilde{x}^\top \tilde{x}+\gamma^2 \tilde{u}^\top \tilde{u} \nonumber\\
      & \leq -\epsilon \tilde{x}^\top \tilde{x}+\gamma^2 \tilde{u}^\top \tilde{u}.
\end{align}

Let us focus now on the second term in the right-hand side of \eqref{eq:dot_V_1}. Since $f$ is continuously differentiable on $D_x\times D_u$ (see Assumption~\ref{assum:f}), we get, considering \eqref{eq:mean_value_f}, that for any $\ell_1>0$ there exists $0<r_x'\leq r_x$ and $0<r_u'\leq r_u$ such that
\begin{align*}
  \Vert \phi(t,\tilde{x},\tilde{u})\Vert_2\leq \ell_1\left(\Vert \tilde{x}\Vert_2+\Vert \tilde{u}\Vert_2 \right),~\forall \Vert\tilde{x}\Vert_2<r_x',~\Vert \tilde{u}\Vert_2<r_u'.  
\end{align*}
Then,
\begin{align}\label{eq:bound_2nd_term_dotV}
    2\tilde{x}^\top P(t)\phi(t,\tilde{x},\tilde{u}) & \leq 2c_2\ell_1 \Vert \tilde{x}\Vert_2 (\Vert \tilde{x}\Vert_2 +\Vert \tilde{u}\Vert_2),\nonumber\\
    & = 2c_2\ell_1 \Vert \tilde{x}\Vert_2^2 + 2c_2\ell_1\Vert \tilde{x} \Vert_2\Vert \tilde{u}\Vert_2 \nonumber\\
    &  \le 3c_2\ell_1\|\tilde{x}\|_2^2 + c_2\ell_1\|\tilde{u}\|_2^2,
\end{align}
where we have used the fact that $P(t)\preceq c_2 I_n$ as well as Young's inequality on the term $\Vert \tilde{x} \Vert_2\Vert \tilde{u}\Vert_2$. 

In view of \eqref{eq:bound_calQ_1} and \eqref{eq:bound_2nd_term_dotV}, we obtain that the  following inequality holds for all $t\geq 0$ and for all $\Vert \tilde{x}\Vert_2\leq r_x'$ and $\Vert \tilde{u}\Vert_2 \leq r_u'$:
\begin{align}\label{eq:bound_dot_V_2}
  \dot V & \leq   -\theta \Vert \tilde{x}\Vert_2^2 +(\gamma ^2+c_2\ell_1)\Vert \tilde{u}\Vert_2^2,
\end{align}
where
\begin{equation}\label{eq:theta}
\theta=\epsilon-3c_2\ell_1.
\end{equation}
Let $\ell_1>0$ be small enough  so  that $\theta>0$ and let $0<\kappa<1$ be arbitrary.  Then, \eqref{eq:bound_dot_V_2} implies the following:
\begin{align}\label{eq:bound_dot_V_3}
   \dot V & \leq -(1-\kappa)\theta \Vert \tilde{x}\Vert_2^2 -\kappa \theta \Vert \tilde{x}\Vert_2^2 \nonumber\\
   & \phantom{=}+ \underbrace{(\gamma^2+c_2\ell_1)(r_u')^2}_{:=\Psi},
\end{align}
 where we have used the fact that $\Vert \tilde{u}\Vert_2\leq r_u'$. It follows from~\eqref{eq:bound_dot_V_3} that
\begin{align}\label{eq:bound_dot_V_4}
    \dot V\leq -(1-\kappa)\theta \Vert \tilde{x}\Vert_2^2,~~ \forall \Vert \tilde{x}\Vert_2\geq \sqrt{\frac{\Psi}{\kappa \theta}}=:\mu>0. 
\end{align}
In view of \eqref{eq:bounds_on_V} and  \eqref{eq:bound_dot_V_4}, we can conclude through \cite[Theorem~4.18]{khalil} that any solution $\tilde{x}$ of the error dynamics \eqref{eq:error_dynamics}, starting sufficiently close from the origin, is uniformly bounded for all $t\geq t_0$. To see that  local uniform internal stability of \eqref{eq:error_dynamics} with respect to the origin holds, it is sufficient to note from \eqref{eq:bound_dot_V_4} that if $\tilde{u}\equiv 0$, then $\dot V\leq -(1-\kappa)\theta \Vert \tilde{x}\Vert_2^2$
 for all $\Vert \tilde{x}\Vert_2>0$, as $\mu$ could be taken as zero, which directly implies local uniform asymptotic stability of the origin by \cite[Theorem~4.8]{khalil}. 

Having established the local uniform boundedness of $\tilde{x}$, we note now that since $\tilde{y}$  depends continuously on $\tilde{x}$, we find that the output error $\tilde{y}$ is also locally uniformly bounded. Indeed, by recalling that $h$ is continuously differentiable on the domain $D_x$ (see Assumption~\ref{assum:f}), we get, considering  \eqref{eq:mean_value_h},  that for any $\ell_2>0$ there exists $ 0<r_x''\leq r_x$ such that
\begin{equation}
  \Vert \psi(t,\tilde{x})\Vert_2\leq   \ell_2 \Vert \tilde{x} \Vert_2,\quad \forall \Vert \tilde{x}\Vert_2 < r_x''.
\end{equation}
Moreover, since $x^\star(t)\in \text{int}\{D_x\}$ for all $t\geq t_0$, we have that
\begin{align*}
\Vert C(t)\Vert_2 \leq k_{C}    ,
\end{align*}
for some $k_C\geq 0$. Then, $\Vert \tilde{y}\Vert_2 \leq k_C \Vert \tilde{x}\Vert_2 +\ell_2 \Vert \tilde{x}\Vert_2^2$ for all $\Vert \tilde{x}\Vert_2\leq r_x''$. Then,  the error dynamics \eqref{eq:error_dynamics} are locally BIBO. We conclude then that the nonlinear system \eqref{eq:system} is locally harmonically stable with 
  respect to the nominal periodic motion $(u^\star(t),x^\star(t), y^\star(t))$ in the sense of Definition~\ref{def:HS}.

\section{Proof of Theorem~\ref{thm:practical_harmonic_stab}}\label{appendix_2}

Assume that all the conditions of the theorem hold and recall that the real-valued HSS dynamics \eqref{eq:HSS_real_stacked} are equivalent to their complex counterpart \eqref{eq:HSS_finite_complex}. Define $\bm{\cP}:= (P_r + jP_i)$. Then $\bar{\cP} = \bar{\cP}^\top \succ 0$ is equivalent to $\bm{\cP} = \bm{\cP}^\star \succ 0$. Now, utilizing \eqref{eq:HSS_real_img_state_decomposition} and \eqref{eq:HSS_real_img_matrices_decomposition}, it is possible to verify  that the  inequality \eqref{eq:HSS_BR_LMI} is equivalent to the following one:
\begin{equation} \label{eq:HSS_BR_LMI_finite_complex} 
\begin{bmatrix} (\bm{\mathcal A}-\bm{\mathcal N})^\star \bm{\mathcal P}
+ \bm{\mathcal P}(\bm{\mathcal A}-\bm{\mathcal N}) 
+ \bm{\mathcal C}^\star \bm{\mathcal C}  +\epsilon I_N & \bm{\mathcal P} \bm{\mathcal B} \\ 
\bm{\mathcal B}^\star \bm{\mathcal P} & -\gamma^2 I_N
\end{bmatrix}
\preceq 0.
\end{equation}
Since  \eqref{eq:HSS_BR_LMI} holds  for any $N \ge N_0$, taking the limit $N \to \infty$ implies that \eqref{eq:HSS_BR_LMI_finite_complex} is satisfied for $N \to \infty$, leading to the following infinite-dimensional matrix inequality 
\begin{equation} \label{eq:HSS_BR_LMI_infinite} 
\begin{bmatrix} (\mathcal A-\mathcal N)^\star \mathcal P
+ \mathcal P(\mathcal A-\mathcal N) 
+ \mathcal C^\star \mathcal C  +\epsilon \cI & \mathcal P \mathcal B \\ 
\mathcal B^\star \mathcal P & -\gamma^2 \cI
\end{bmatrix}
\preceq 0 ,
\end{equation}
where $\cI$ is the infinite-dimensional identity operator and $\cP = \cP^\star \succ 0$ is an infinite-dimensional Toeplitz operator.
Using the Schur complement lemma, \eqref{eq:HSS_BR_LMI_infinite} is equivalent to
\begin{align*}
   (\mathcal A-\mathcal N)^\star \mathcal P
+ \mathcal P(\mathcal A-\mathcal N) 
+ \mathcal C^\star \mathcal C + \frac{1}{\gamma^2}\cP \cB \cB^\star \cP +\epsilon \cI \preceq 0.
\end{align*}
Then, there exists a Toeplitz Hermitian positive-semidefinite operator $\cQ$ such that 
\begin{align}\label{eq:proof_appendix_2_Schur_infinite_dim}
&   (\mathcal A-\mathcal N)^\star \mathcal P
+ \mathcal P(\mathcal A-\mathcal N) 
+ \mathcal C^\star \mathcal C+\cdots \nonumber\\
&\phantom{=}\cdots + \frac{1}{\gamma^2}\cP \cB \cB^\star \cP +\epsilon \cI + \cQ = 0.
\end{align}
Let $\cT^{-1}$ be the inverse Toeplitz operator mapping $\cP$ into a $T$-periodic matrix function $t\mapsto P(t)$, and analogously mapping $\mathcal{A},\mathcal{B},\mathcal{C},\mathcal{Q}$ into the $T$-periodic matrix functions $A(t),B(t),C(t),Q(t)$, with $Q(t)\succeq 0$ for all $t\geq 0$. Then, invoking  \cite[Theorem 6]{Blin2022}, it is possible to claim that \eqref{eq:proof_appendix_2_Schur_infinite_dim} is equivalent to the following matrix equation:
\begin{align*}
    &\dot{P}(t) + A^\top(t)P(t) + P(t)A(t) + C^\top(t)C(t)+\cdots \\
    &\cdots + \frac{1}{\gamma^2}P(t)B(t)B^\top(t)P(t) + \epsilon I + Q(t) = 0,
\end{align*}
which, by the Schur complement lemma, implies the feasibility of the time-varying matrix inequality~\eqref{eq:P(t)}. Therefore, by Proposition~\ref{thm:lin_to_local_iss}, the nonlinear system~\eqref{eq:system} is locally harmonically stable with respect to $(u^\star(t),x^\star(t),y^\star(t))$ in the sense of Definition~\ref{def:HS}.

\bibliographystyle{IEEEtran}
\bibliography{references}

\vfill
\end{document}